\documentclass[10pt]{article}

\makeatletter

\newif\iftwodates
\twodatesfalse

\renewcommand\maketitle{\begin{titlepage}%
  \pagenumbering{Alph}
  \let\footnotesize\small
  \let\footnoterule\relax
  \let\footnote\thanks
  \null\vfil
  \vskip 30\p@
  \begin{center}%
    {\LARGE \bf \@title \par}%
    \vskip 3em%
    {\large
     \lineskip .75em%
     \begin{tabular}[t]{c}%
       \@author
     \end{tabular}\par}%
     \vskip 1.5em%
  \end{center}\par
  \vfill
  \begin{center}
    \raisebox{1.5cm}{\includegraphics[width=0.58\textwidth]%
      {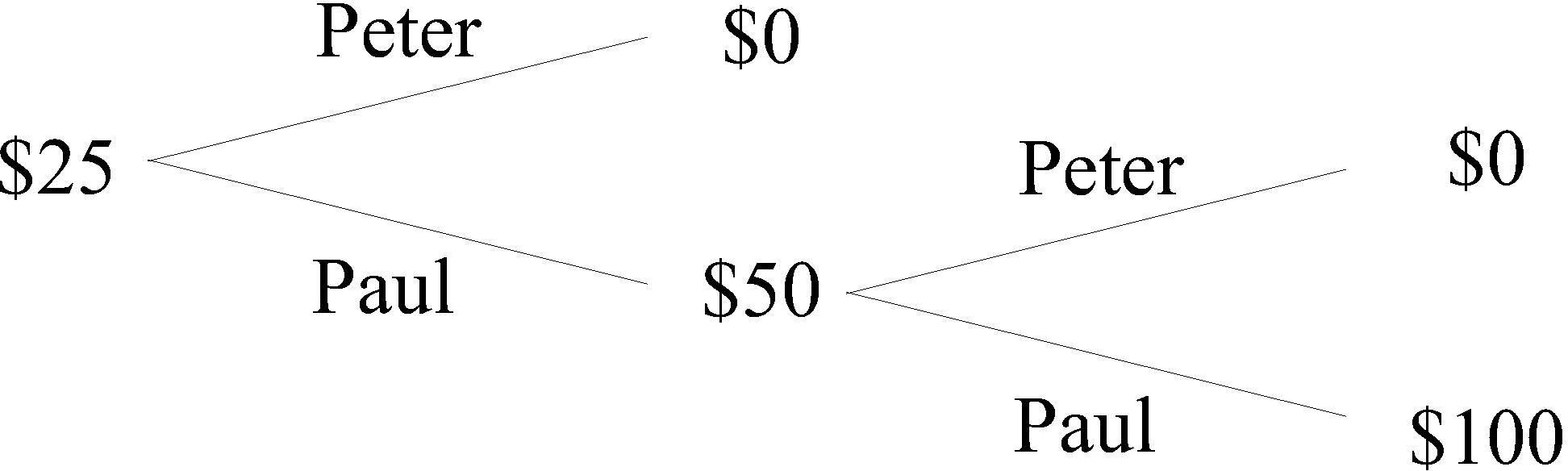}}%
    \hskip 3em%
    \includegraphics[width=0.29\textwidth]%
      {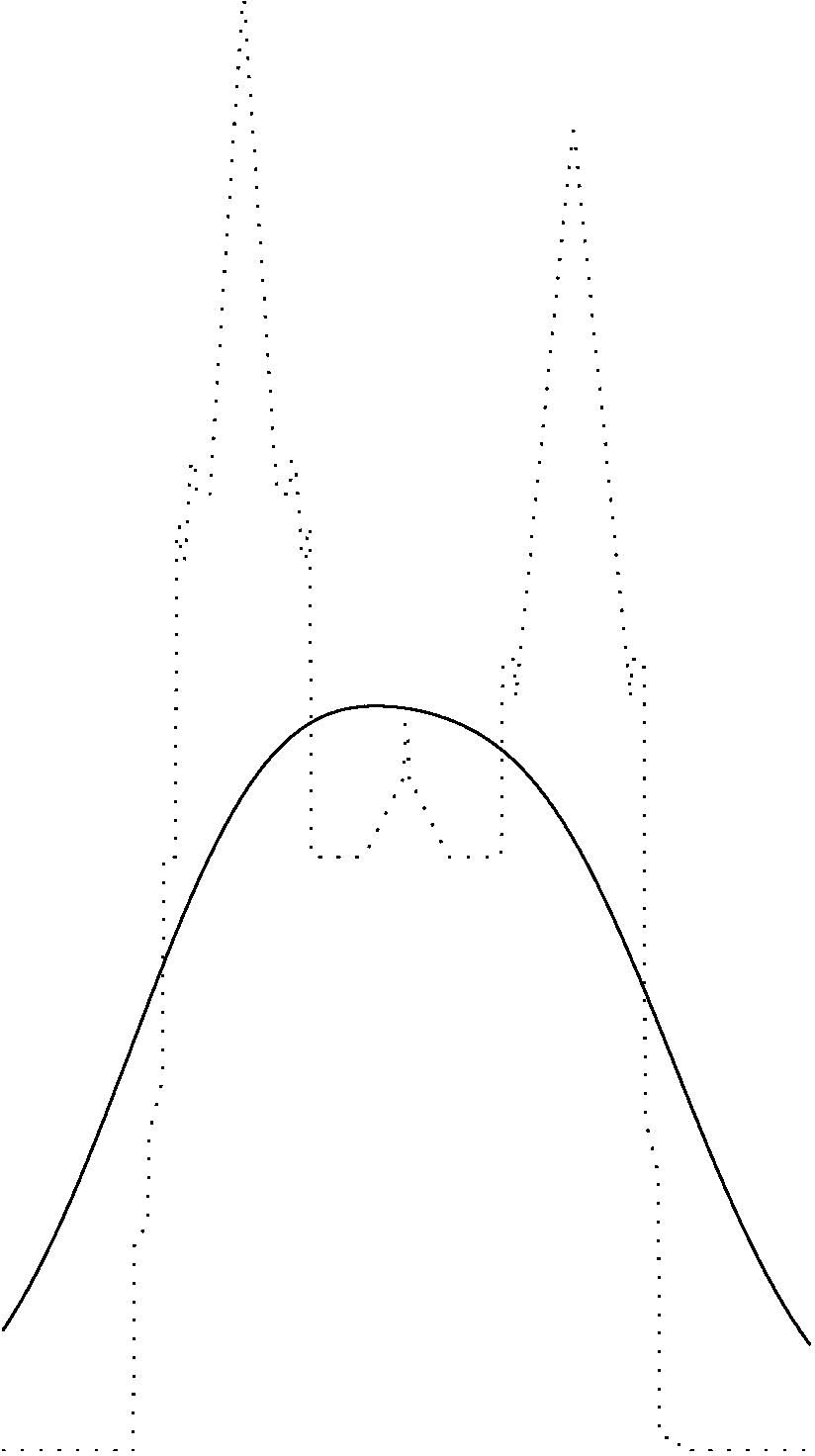}%
  \end{center}
  \@thanks
  \vfill
  \begin{center}
    {\large \bf The Game-Theoretic Probability and Finance Project}
  \end{center}
  \begin{center}
    {\large Working Paper \#\No}
  \end{center}
  \begin{center}
    {\iftwodates\large First posted \firstposted.
    Last revised \@date.\else\large\@date\fi}
  \end{center}
  \begin{center}
    Project web site:\\
    http://www.probabilityandfinance.com
  \end{center}
  \end{titlepage}%
  \setcounter{footnote}{0}%
  \global\let\thanks\relax
  \global\let\maketitle\relax
  \global\let\@thanks\@empty
  \global\let\@author\@empty
  \global\let\@date\@empty
  \global\let\@title\@empty
  \global\let\title\relax
  \global\let\author\relax
  \global\let\date\relax
  \global\let\and\relax
}

\renewenvironment{abstract}{%
  \titlepage\pagenumbering{roman}
  \null\vfil
  \@beginparpenalty\@lowpenalty
  \begin{center}%
    \Large \bfseries \abstractname
    \@endparpenalty\@M
  \end{center}}%
  {\par\vfill\tableofcontents\thispagestyle{empty}\endtitlepage
  \pagenumbering{arabic}}

\renewenvironment{thebibliography}[1]
  {\section*{\refname}%
  \addcontentsline{toc}{section}{\refname}%
  \@mkboth{\MakeUppercase\refname}{\MakeUppercase\refname}%
  \list{\@biblabel{\@arabic\c@enumiv}}%
    {\settowidth\labelwidth{\@biblabel{#1}}%
    \leftmargin\labelwidth
    \advance\leftmargin\labelsep
    \@openbib@code
    \usecounter{enumiv}%
    \let\p@enumiv\@empty
    \renewcommand\theenumiv{\@arabic\c@enumiv}}%
    \sloppy
    \clubpenalty4000
    \@clubpenalty \clubpenalty
    \widowpenalty4000%
    \sfcode`\.\@m}
    {\def\@noitemerr
    {\@latex@warning{Empty `thebibliography' environment}}%
  \endlist}

\makeatother

\usepackage{amsmath,amsfonts,latexsym,natbib,graphicx,url}
\usepackage[pdfpagemode=UseNone,pdfstartview=FitH]{hyperref}

\newcommand{\st}{\mathop{|}}
\newcommand{\vex}{\mathop{\mathrm{vex}}\nolimits}
\newcommand{\Variation}{\mathop{\mathbf{var}}\nolimits}
\newcommand{\ns}{{}^*\!}

\newcommand{\bbbh}{\mathbb{H}}
\newcommand{\bbbr}{\mathbb{R}}
\newcommand{\bbbt}{\mathbb{T}}
\newcommand{\bbbn}{\mathbb{N}}

\newcommand{\III}{\mathcal{I}}
\newcommand{\RRR}{\mathcal{R}}
\newcommand{\SSS}{\mathcal{S}}

\newcommand{\UUU}{\mathcal{U}}

\newcommand{\OOO}{\mathcal{O}}

\newcommand{\standard}{\mathop{\mathrm{st}}}

\newtheorem{theorem}{Theorem}
\newtheorem{corollary}{Corollary}

\newtheorem{Remark}{Remark}

\newenvironment{proof}
  {\trivlist \item[\hskip\labelsep\textbf{Proof}]}
  {\endtrivlist}

\def\squareforqed{\rule{.3em}{1.5ex}}
\def\qed{\ifmmode{\lefteqn{\;\;\squareforqed}}\else{\unskip\nobreak\hfil
\penalty50\hskip1em\null\nobreak\hfil\squareforqed
\parfillskip=0pt\finalhyphendemerits=0\endgraf}\fi}

\newlength{\IndentI}
\newlength{\IndentII}
\newlength{\IndentIII}
\newlength{\WidthI}
\newlength{\WidthII}
\newlength{\WidthIII}
\title{A game-theoretic derivation of the $\sqrt{dt}$ effect}
\author{Vladimir Vovk and Glenn Shafer}
\newcommand{\No}{5}
\twodatestrue
\newcommand{\firstposted}{January 27, 2003}

\begin{document}

\setcounter{tocdepth}{1}
\maketitle

\begin{abstract}
  We study the origins of the $\sqrt{dt}$ effect in finance and SDE.
  In particular, we show, in the game-theoretic framework, that market volatility is a consequence
  of the absence of riskless opportunities for making money
  and that too high volatility is also incompatible with such opportunities.
  More precisely,
  riskless opportunities for making money arise
  whenever a traded security has fractal dimension below or above that of the Brownian motion
  and its price is not almost constant and does not become extremely large.
  This is a simple observation known in the measure-theoretic mathematical finance.
  At the end of the article we also consider the case of non-zero interest rate.

  This version of the article was essentially written in March 2005
  but remains a working paper.
\end{abstract}

\section{Introduction}

The main result of this article is that
high market volatility is a consequence of the absence of riskless opportunities for making money.
Versions of this proposition were proven within the standard continuous-time framework
by \citet{rogers:1997}
(see also the references therein),
\citet{delbaen/schachermayer:1994},
etc.

In \S\ref{sec:ContFin} we prove a simple result using nonstandard analysis
saying that if a traded security is not sufficiently volatile
and not too close to being a constant,
this can be used for making money without risk;
in the appendix to this article we explain how the informal language of \S\ref{sec:ContFin}
can be replaced by a formal argument using the ultraproduct construction
described in \citet{shafer/vovk:2001}.
In the following \S\ref{sec:absolute} and \S\ref{sec:relative}
we give messier finitary forms of this result
in a realistic, discrete-time setting.
Results of our preliminary empirical studies are reported in \S\ref{sec:empirical}.

In \S\ref{sec:interest}, we remove the assumption of zero interest rate.
Our proof techniques are elementary and well-known;
see, e.g.,
Cheridito (\citeyear{cheridito:2001}, \citeyear{cheridito:2002}).
(Although the techniques are general,
the results are typically stated for very narrow classes of processes:
fractional Brownian motion with drift and exponential fractional Brownian motion with drift
in Cheridito \citeyear{cheridito:2001}, \citeyear{cheridito:2002}.)

In \S\ref{sec:alternative} we briefly discuss a modification of the Market Protocol of \S\ref{sec:ContFin}
that allows more natural statements of the results of \S\ref{sec:ContFin}.

\section{Continuous-time result in the financial protocol}\label{sec:ContFin}

We use the notation of \citet{shafer/vovk:2001}.  
In particular, $\Delta f_n := f_n - f_{n-1}$, while $df_n:=f_{n+1}-f_n$.
The basic framework is that of Chapter~11:
the interval $T$ is split into an infinitely large number $N$ of subintervals etc.

\bigskip

\noindent
\textsc{The Market Protocol}

\noindent
\textbf{Players:} Investor, Market

\noindent
\textbf{Protocol:}

  \parshape=6
  \IndentI  \WidthI
  \IndentI  \WidthI
  \IndentI  \WidthI
  \IndentII \WidthII
  \IndentII \WidthII
  \IndentII \WidthII
  \noindent
  $\III_0:=1$.\\
  Market announces $S_0\in\bbbr$.\\
  FOR $n=1,2,\dots,N$:\\
      Investor announces $M_n\in\bbbr$.\\
      Market announces $S_n\in\bbbr$.\\
      $\III_n := \III_{n-1} + M_n \Delta S_n$.

\noindent
\textbf{Additional Constraint on Market:}
Market must ensure that $S$ is continuous.

\bigskip

The definition of zero game-theoretic probability is given on pp.~340--341 of \citet{shafer/vovk:2001}:
an event $E$ has zero game-theoretic probability\label{p:def} if for any $K$ there exists a strategy
that, when started with $1$, does not risk bankruptcy
and finishes with capital at least $K$ when $E$ happens.

We start with a result showing that too low volatility
gives opportunities for making money.
\begin{theorem}\label{thm:nonstandardlow}
  For any $\delta>0$, the event
  \[
    \vex S < 2
    \;\&\;
    \sup_t
    \left|
      S(t)-S(0)
    \right|
    >
    \delta
  \]
  has game-theoretic probability zero.
\end{theorem}
The condition $\vex S < 2$ means that $S$ is less volatile than the Brownian motion
and
$
  \sup_t
  \left|
    S(t)-S(0)
  \right|
  >
  \delta
$
means that $S$ should not be almost constant.
\begin{proof}
  This proof is a simple modification of Example~3 in \citet{shiryaev:1999}, p.~658,
  and a proof in
  Cheridito (\citeyear{cheridito:2001}, \citeyear{cheridito:2002}).
  It is given in the usual style of~\citet{shafer/vovk:2001};
  in the appendix we will provide additional details.

  Assume, without loss of generality, that $S(0)=0$
  (if this is not true, replace $S(t)$ by $S(t)-S(0)$).
  Consider the strategy $M_n:=2CS_n$, where $C$ is a large positive constant.
  With our usual notation $df_n:=f_{n+1}-f_n$, we have
  \[
    d\III_n
    =
    2CS_ndS_n
    =
    C
    \left(
      d(S_n^2) - (dS_n)^2
    \right)
  \]
  and, therefore,
  \begin{equation}\label{eq:change1}
    \III_n - \III_0
    =
    C S_n^2
    -
    C
    \sum_{i=0}^{n-1}
    (dS_i)^2
    \approx
    C S_n^2.
  \end{equation}
  If this strategy starts with $1$,
  the capital at each step $n$ will be nonnegative.
  Stopping playing at the first step when $|S_n|>\delta$,
  we make sure that $\III_N\ge C\delta^2$,
  which can be made arbitrarily large by taking a large $C$.
  \qed
\end{proof}

As the proof shows,
the condition $\vex S<2$ of the theorem can be replaced
by the weaker $\Variation_S(2) = 0$.

Now we complement Theorem~\ref{thm:nonstandardlow}
with a result dealing with too high volatility.
\begin{theorem}\label{thm:nonstandardhigh}
  For any $D>0$, the event
  \begin{equation}\label{eq:event0}
    \vex S > 2
    \;\&\;
    \sup_t
    \left|
      S(t)-S(0)
    \right|
    <
    D
  \end{equation}
  has game-theoretic probability zero.
\end{theorem}

\begin{proof}
  This proof is a simple modification of a proof in
  Cheridito (\citeyear{cheridito:2001}, \citeyear{cheridito:2002}).
  We again assume $S(0)=0$.

  Consider the strategy $M_n:=-2D^{-2}S_n$.
  Now we have
  \[
    d\III_n
    =
    -2D^{-2}S_ndS_n
    =
    D^{-2}
    \left(
      (dS_n)^2 - d(S_n^2)
    \right)
  \]
  and, therefore,
  \begin{equation}\label{eq:change2}
    \III_n - \III_0
    =
    D^{-2}
    \sum_{i=0}^{n-1}
    (dS_i)^2
    -
    D^{-2} S_n^2
    \ge
    D^{-2}
    \sum_{i=0}^{n-1}
    (dS_i)^2
    -
    1
  \end{equation}
  before $|S_n|$ reaches $D$.
  If this strategy starts with $1$ and stops playing as soon as $S_n$ reaches $D$,
  the capital at each step $n$ will be nonnegative
  and, if event~(\ref{eq:event0}) occurs,
  $\III_N\ge D^{-2}\Variation_S(2)$ will be infinitely large.
  \qed
\end{proof}

As before,
the condition $\vex S>2$ can be replaced by $\Variation_S(2) = \infty$.

If $S$ is a stock price, it cannot become negative,
which allows us to strengthen the conclusion of Theorem~\ref{thm:nonstandardhigh}.
\begin{corollary}\label{cor:nonstandardhigh}
  The event $\vex S > 2$ has game-theoretic probability zero (provided $S\ge0$).
\end{corollary}
\begin{proof}
  Let $K$ be the constant from the definition of zero game-theoretic probability
  (p.~\pageref{p:def}).
  The required strategy is the 50/50 mixture of the following 2 strategies:
  the strategy of Theorem~\ref{thm:nonstandardhigh} corresponding to $D:=2K$
  and the buy-and-hold strategy that recommends buying 1 share of $S$ at the outset.
  If $\sup_t|S(t)-S(0)|<2K$, the first strategy will make Investor rich;
  otherwise, the second will.
  \qed
\end{proof}

\section{Absolute finitary results}\label{sec:absolute}

The protocol for this section is:

\bigskip

\noindent
\textsc{The Absolute Market Protocol}

\noindent
\textbf{Players:} Investor, Market

\noindent
\textbf{Protocol:}

  \parshape=6
  \IndentI  \WidthI
  \IndentI  \WidthI
  \IndentI  \WidthI
  \IndentII \WidthII
  \IndentII \WidthII
  \IndentII \WidthII
  \noindent
  $\III_0:=1$.\\
  Market announces $S_0\in\bbbr$.\\
  FOR $n=1,2,\dots,N$:\\
      Investor announces $M_n\in\bbbr$.\\
      Market announces $S_n\in\bbbr$.\\
      $\III_n := \III_{n-1} + M_n \Delta S_n$.

\bigskip

Now $N$ is a usual positive integer number and there are no \emph{a priori} constraints on Market.
The following two results
are the ``absolute'' finitary versions
of Theorems~\ref{thm:nonstandardlow} and~\ref{thm:nonstandardhigh},
respectively.
\begin{theorem}\label{thm:absolute1}
  Let $\epsilon$ and $\delta$ be two positive numbers.
  If Market is required to satisfy
  \[
    \sum_{i=1}^N
    (\Delta S_i)^2
    \le
    \epsilon,
  \]
  the game-theoretic probability of the event
  \begin{equation}\label{eq:event1}
    \max_{n=1,\dots,N}
    |S_n - S_0|
    \ge
    \delta
  \end{equation}
  is at most $\epsilon / \delta^2$.
\end{theorem}
\begin{proof}
  Assume, without loss of generality, that $S_0=0$
  (replace $S_n$ by $S_n-S_0$ if not).
  Take the same strategy $M_n:=2CS_n$
  as in Theorem~\ref{thm:nonstandardlow},
  but now $C=1/\epsilon$.
  From (\ref{eq:change1}) we obtain
  \[
    \III_n - \III_0
    =
    \frac{1}{\epsilon}
    S_n^2
    -
    \frac{1}{\epsilon}
    \sum_{i=0}^{n-1}
    (dS_i)^2
    \ge
    \frac{1}{\epsilon}
    S_n^2
    -
    1,
  \]
  i.e., if the strategy starts with $1$,
  \[
    \III_n
    \ge
    \frac{1}{\epsilon}
    S_n^2.
  \]
  This shows that $\III_n$ is never negative;
  stopping at the step $n$ where $|S_n|\ge\delta$,
  we make sure that $\III_N\ge\delta^2/\epsilon$ when~(\ref{eq:event1}) happens.
  \qed
\end{proof}
\begin{theorem}\label{thm:absolute2}
  Let $\epsilon$ and $D$ be two positive numbers.
  If Market is required to satisfy
  \[
    \max_{n=1,\dots,N}
    |S_n - S_0|
    \le
    D,
  \]
  the upper game-theoretic probability of the event
  \begin{equation}\label{eq:event1high}
    \sum_{i=1}^N
    (\Delta S_i)^2
    \ge
    \frac{D^2}{\epsilon}
  \end{equation}
  is at most $\epsilon$.
\end{theorem}
\begin{proof}
  Assume, without loss of generality, that $S_0=0$
  (replace $S_n$ by $S_n-S_0$ if not).
  Take the same strategy $M_n:=-2D^{-2}S_n$
  as in Theorem~\ref{thm:nonstandardhigh}.
  From (\ref{eq:change2}) we can see that $\III_n$ is never negative
  and that
  \[
    \III_N
    =
    D^{-2}
    \sum_{i=1}^N
    (\Delta S_i)^2
    \ge
    \frac{1}{\epsilon}
  \]
  when the event (\ref{eq:event1high}) happens.
  \qed
\end{proof}

\section{Relative finitary result}\label{sec:relative}

Now we change our protocol to:

\bigskip

\noindent
\textsc{The Relative Market Protocol}

\noindent
\textbf{Players:} Investor, Market

\noindent
\textbf{Protocol:}

  \parshape=6
  \IndentI  \WidthI
  \IndentI  \WidthI
  \IndentI  \WidthI
  \IndentII \WidthII
  \IndentII \WidthII
  \IndentII \WidthII
  \noindent
  $\III_0:=1$.\\
  Market announces $S_0>0$.\\
  FOR $n=1,2,\dots,N$:\\
      Investor announces $M_n\in\bbbr$.\\
      Market announces $S_n>0$.\\
      $\III_n := \III_{n-1} + M_n \Delta S_n$.

\bigskip

As in the previous section,
$N$ is a standard positive integer number.
Define a nonnegative function $\beta$ by
\[
  \frac12 \beta(x)
  =
  x - \ln(1+x);
\]
so for small $|x|$, $\beta(x)$ behaves as $x^2$.
The following result is the ``relative'' finitary version of Theorem~\ref{thm:nonstandardlow};
it uses the versions
\[
  \sum_{i=0}^{N-1}
  (d \ln S_i)^2
\]
and
\[
  \sum_{i=0}^{N-1}
  \beta
  \left(
    \frac{d S_i}{S_i}
  \right)
\]
of the 2-variation
\[
  \sum_{i=0}^{N-1}
  \left(
    \frac{d S_i}{S_i}
  \right)^2
\]
of $S$.
(We refrain from giving a similar version of Theorem~\ref{thm:nonstandardhigh}:
such a version would be less interesting from the empirical point of view,
because, as explained in the following section,
the usual expectation is that $\bbbh>1/2$.)
\begin{theorem}\label{thm:relative}
  Let $\epsilon$, $\delta$ and $\gamma$ be three positive numbers.
  If Market is required to satisfy
  \[
    \sum_{i=0}^{N-1}
    (d \ln S_i)^2
    \le
    \epsilon,
  \quad
    \sum_{i=0}^{N-1}
    \beta
    \left(
      \frac{d S_i}{S_i}
    \right)
    \le
    \epsilon
  \]
  and
  \[
    \min_n \ln S_n \ge -\gamma,
  \]
  the game-theoretic probability of the event
  \begin{equation}\label{eq:event2}
    \max_{n=1,\dots,N}
    |\ln(S_n/S_0)|
    \ge
    \delta
  \end{equation}
  is at most $(1+\gamma)\epsilon / \delta^2$.
\end{theorem}
\begin{proof}
  Assume, without loss of generality, that $S_0=1$
  (replace $S_n$ by $S_n/S_0$ if not).
  Since the proof is now slightly more complicated than that in the previous section,
  we first outline its idea.
  Roughly speaking,
  our goal will be to maintain $\III_n$ close to $(\ln S_n)^2$
  (in the previous sections it was to maintain $\III_n$ close to $(S_n)^2$).
  To find a strategy that will achieve this,
  we notice that
  \begin{eqnarray}\label{eq:prel1}
    d
    \left(
      \ln^2 S_n
    \right)
    &=&
    (2 \ln S_n)
    (d \ln S_n)
    +
    (d \ln S_n)^2
  \\\label{eq:prel2}
    &=&
    (2 \ln S_n)
    \ln
    \left(
      1 + \frac{dS_n}{S_n}
    \right)
    +
    (d \ln S_n)^2
  \\\label{eq:prel3}
    &=&
    (2 \ln S_n)
    \frac{dS_n}{S_n}
    -
    (\ln S_n)
    \beta
    \left(
      \frac{dS_n}{S_n}
    \right)
    +
    (d \ln S_n)^2.
  \end{eqnarray}
  We can see that a suitable strategy is
  \[
    M_n
    :=
    2C
    \frac{\ln S_n}{S_n}
  \]
  for some $C$
  (chosen so that to make sure that the capital process is nonnegative;
  eventually we will take $C=1/((1+\gamma)\epsilon)$).
  Expressing $2 (\ln S_n)(dS_n)/S_n$ from the equality
  between the extreme terms of the chain~(\ref{eq:prel1})--(\ref{eq:prel3}),
  we obtain for the strategy $M_n$:
  \[
    d\III_n
    =
    2C
    \frac{\ln S_n}{S_n}dS_n
    =
    Cd\ln^2S_n
    +
    C\ln S_n
    \beta
    \left(
      \frac{dS_n}{S_n}
    \right)
    -
    C(d\ln S_n)^2;
  \]
  therefore,
  \begin{eqnarray}\label{eq:change2a}
    \III_n - \III_0
    &=&
    C \ln^2S_n
    +
    C
    \sum_{i=0}^{n-1}
    (\ln S_i)
    \beta
    \left(
      \frac{dS_i}{S_i}
    \right)
    -
    C
    \sum_{i=0}^{n-1}
    (d\ln S_i)^2\\\label{eq:change2b}
    &\ge&
    C \ln^2S_n
    -
    C\gamma\epsilon
    -
    C
    \epsilon.
  \end{eqnarray}
  Starting from $\III_0=1$, it is safe to take $C:=1/((1+\gamma)\epsilon)$
  (this removes possibility of bankruptcy),
  in which case~(\ref{eq:change2a})--(\ref{eq:change2b}) becomes
  \[
    \III_n
    \ge
    \frac{1}{(1+\gamma)\epsilon}
    \ln^2S_n.
  \]
  Stopping at the step $n$ with $|\ln S_n|\ge\delta$
  ensures $\III_N\ge\delta^2/((1+\gamma)\epsilon)$ when~(\ref{eq:event2}) happens.
  \qed
\end{proof}

\section{Empirical studies}\label{sec:empirical}

The empirical studies reported in this section
are closely connected to the so called $\RRR/\SSS$-analysis
(see \citealt{shiryaev:1999}, \S4a).
The results we report here assume zero interest rate,
and so are of limited interest;
further empirical studies are needed.

First we consider the absolute setting,
although the usual definitions as given in~\citet{shiryaev:1999}
are ``relative''.
Denote
\[
  \RRR^{{\rm abs}}_N
  :=
  \max_{i=1,\dots,N}
  |S_n - S_0|,
  \quad
  \left(
    \SSS^{{\rm abs}}_N
  \right)^2
  :=
  \frac1N
  \sum_{i=1}^N
  (\Delta S_i)^2.
\]
Suppose that we believe, for some reason,
that we are going to have $\SSS^{{\rm abs}}_N \le \sigma$ and $\RRR^{{\rm abs}}_N \ge \delta$;
therefore, $\delta$ plays the same role as in Theorem~\ref{thm:absolute1}
and $\sigma$ plays the role of $\sqrt{\epsilon/N}$.
So from Theorem~\ref{thm:absolute1} we obtain that
we will be able to multiply our capital
$\delta^2/\epsilon = (\delta/\sigma)^2/N$-fold.
For another variant of the definitions of $\RRR_N$ and $\SSS_N$
(as given in \citealt{shiryaev:1999}, (14) on p.~371; see below)
one usually has
\[
  \frac{\RRR_N}{\SSS_N}
  \sim
  c N^{\bbbh}
\]
with $\bbbh$ considerably larger than $1/2$.
If our guesses $\delta$ and $\sigma$ are not too far off,
we can hope to increase our initial capital
by a factor of order $N^{2\bbbh-1}$.

In the ``relative'' setup, define
\[
  \RRR^{{\rm rel}}_N
  :=
  \max_{n=1,\dots,N}
  \left|
    \ln\frac{S_n}{S_0}
  \right|,
  \quad
  \left(
    \SSS^{{\rm abs}}_N
  \right)^2
  :=
  \frac1N
  \left(
    \sum_{i=0}^{N-1}
    \beta
    \left(
      \frac{dS_i}{S_i}
    \right)
    \;\bigvee\;
    \sum_{i=0}^{N-1}
    (d\ln S_i)^2
  \right).
\]
If we believe that we are going to have $\SSS^{{\rm rel}}_N \le \sigma$ and $\RRR^{{\rm rel}}_N \ge \delta$,
we obtain from Theorem~\ref{thm:relative} that
we will be able to multiply our capital by a factor of
\[
  \frac{\delta^2}{(1+\gamma)\epsilon}
  =
  \frac{\delta^2}{(1+\gamma)\sigma^2N}
\]
(where $\epsilon:=\sigma^2 N$);
if one has
\[
  \frac{\RRR_N}{\SSS_N}
  \sim
  c N^{\bbbh}
\]
and our guesses $\delta$ and $\sigma$ are not too far off,
we can again hope to increase our initial capital
by a factor of order
\begin{equation}\label{eq:H}
  N^{2\bbbh-1}.
\end{equation}

Some experimental results are given in
Shiryaev (\citeyear{shiryaev:1999}, \S4.4),
but we cannot use them directly,
since the standard definitions of $\RRR/\SSS$ analysis are different from ours
(the main difference being that the standard definitions are centered).
Those results, however, suggest that typically $\bbbh>0.5$,
which was why we concentrate on this case in our discrete-time analysis and empirical studies.

In our experiments we consider,
instead of $\RRR^{{\rm abs}}_N$ and $\RRR^{{\rm rel}}_N$,
$|S_N - S_0|$ and $|\ln(S_N/S_0)|$, respectively,
the rationale being that security prices typically increase.
This frees us from the need to guess the value of $\delta$ in advance.
Our results are summarized in Tables~\ref{tab:list} and~\ref{tab:results}.

\begin{table}
  \begin{tabular}{llll}
    \hline
      {\footnotesize Security and frequency}
      &
      {\footnotesize Code}
      &
      {\footnotesize Time Period}
      &
      {\footnotesize $N$}\\
    \hline
      {\footnotesize Microsoft stock daily}
      &
      {\footnotesize msft d}
      &
      {\footnotesize 13/03/1986--21/09/2001}
      &
      {\footnotesize $3672$}
    \\
      {\footnotesize IBM stock daily}
      &
      {\footnotesize ibm d}
      &
      {\footnotesize 02/01/1962--21/09/2000}
      &
      {\footnotesize $9749$}
    \\
      {\footnotesize S\&P500 daily}
      &
      {\footnotesize spc d}
      &
      {\footnotesize 04/01/1960--21/09/2000}
      &
      {\footnotesize $10,254$}
    \\
      {\footnotesize Microsoft stock monthly}
      &
      {\footnotesize msft m}
      &
      {\footnotesize March 1986--June 2001}
      &
      {\footnotesize $184$}
    \\
      {\footnotesize IBM stock monthly}
      &
      {\footnotesize ibm m}
      &
      {\footnotesize January 1962--June 2001}
      &
      {\footnotesize $474$}
    \\
      {\footnotesize General Electric stock monthly}
      &
      {\footnotesize ge m}
      &
      {\footnotesize January 1962--June 2001}
      &
      {\footnotesize $474$}
    \\
      {\footnotesize Boeing stock monthly}
      &
      {\footnotesize ba m}
      &
      {\footnotesize January 1970--June 2001}
      &
      {\footnotesize $378$}
    \\
      {\footnotesize Du Pont (E.I.)\ de Nemours stock monthly}
      &
      {\footnotesize dd m}
      &
      {\footnotesize January 1970--June 2001}
      &
      {\footnotesize $378$}
    \\
      {\footnotesize Consolidated Edison stock monthly}
      &
      {\footnotesize ed m}
      &
      {\footnotesize January 1970--June 2001}
      &
      {\footnotesize $378$}
    \\
      {\footnotesize Eastman Kodak stock monthly}
      &
      {\footnotesize ek m}
      &
      {\footnotesize January 1970--June 2001}
      &
      {\footnotesize $378$}
    \\
      {\footnotesize General Motors stock monthly}
      &
      {\footnotesize gm m}
      &
      {\footnotesize January 1970--June 2001}
      &
      {\footnotesize $378$}
    \\
      {\footnotesize Procter and Gamble stock monthly}
      &
      {\footnotesize pg m}
      &
      {\footnotesize January 1970--June 2001}
      &
      {\footnotesize $378$}
    \\
      {\footnotesize Sears/Roebuck stock monthly}
      &
      {\footnotesize s m}
      &
      {\footnotesize January 1970--June 2001}
      &
      {\footnotesize $378$}
    \\
      {\footnotesize AT\&T stock monthly}
      &
      {\footnotesize t m}
      &
      {\footnotesize January 1970--June 2001}
      &
      {\footnotesize $378$}
    \\
      {\footnotesize Texaco stock monthly}
      &
      {\footnotesize tx m}
      &
      {\footnotesize January 1970--June 2001}
      &
      {\footnotesize $378$}
    \\
      {\footnotesize US T-bill monthly}
      &
      {\footnotesize us m}
      &
      {\footnotesize January 1871--June 2001}
      &
      {\footnotesize $1566$}
    \\
      {\footnotesize S\&P500 Total Returns monthly}
      &
      {\footnotesize sp m}
      &
      {\footnotesize January 1871--June 2001}
      &
      {\footnotesize $1566$}
    \\
      {\footnotesize US T-bill yearly}
      &
      {\footnotesize us a}
      &
      {\footnotesize 1871--2000}
      &
      {\footnotesize $130$}
    \\
      {\footnotesize S\&P500 Total Returns yearly}
      &
      {\footnotesize sp a}
      &
      {\footnotesize 1871--2001}
      &
      {\footnotesize $130$}
  \end{tabular}
  \caption{The 19 securities used in our experiments.
    Dates are given in the format dd/mm/yyyy.\label{tab:list}}
\end{table}

In Table~\ref{tab:list} we list the 19 securities for which we conducted experiments.
The number $N$ is the number of trading periods (days, month, or years).

\begin{table}
  \begin{center}
  \begin{tabular}{llllll}
    \hline
      {\footnotesize code}
      &
      {\footnotesize abs factor}
      &
      {\footnotesize rel factor}
      &
      {\footnotesize index}
      &
      {\footnotesize security}
      &
      {\footnotesize min}\\
    \hline
      {\footnotesize msft d}
      &
      {\footnotesize $1.32$}
      &
      {\footnotesize $13.8$}
      &
      {\footnotesize $6.13$}
      &
      {\footnotesize $330$}
      &
      {\footnotesize $-0.0736$}
    \\
      {\footnotesize ibm d}
      &
      {\footnotesize $2.62$}
      &
      {\footnotesize $1.94$}
      &
      {\footnotesize $20.8$}
      &
      {\footnotesize $16.0$}
      &
      {\footnotesize $-0.626$}
    \\
      {\footnotesize spc d}
      &
      {\footnotesize $8.25$}
      &
      {\footnotesize $10.7$}
      &
      {\footnotesize $24.2$}
      &
      {\footnotesize $24.2$}
      &
      {\footnotesize $-0.138$}
    \\
      {\footnotesize msft m}
      &
      {\footnotesize $0.930$}
      &
      {\footnotesize $12.6$}
      &
      {\footnotesize $7.59$}
      &
      {\footnotesize $382$}
      &
      {\footnotesize $0$}
    \\
      {\footnotesize ibm m}
      &
      {\footnotesize $2.08$}
      &
      {\footnotesize $2.22$}
      &
      {\footnotesize $69.4$}
      &
      {\footnotesize $15.8$}
      &
      {\footnotesize $-0.469$}
    \\
      {\footnotesize ge m}
      &
      {\footnotesize $4.57$}
      &
      {\footnotesize $6.99$}
      &
      {\footnotesize $69.4$}
      &
      {\footnotesize $63.4$}
      &
      {\footnotesize $-0.232$}
    \\
      {\footnotesize ba m}
      &
      {\footnotesize $2.35$}
      &
      {\footnotesize $3.25$}
      &
      {\footnotesize $43.5$}
      &
      {\footnotesize $75.5$}
      &
      {\footnotesize $-0.620$}
    \\
      {\footnotesize dd m}
      &
      {\footnotesize $1.66$}
      &
      {\footnotesize $6.72$}
      &
      {\footnotesize $43.5$}
      &
      {\footnotesize $33.0$}
      &
      {\footnotesize $0$}
    \\
      {\footnotesize ed m}
      &
      {\footnotesize $4.21$}
      &
      {\footnotesize $3.76$}
      &
      {\footnotesize $43.5$}
      &
      {\footnotesize $58.0$}
      &
      {\footnotesize $-1.19$}
    \\
      {\footnotesize ek m}
      &
      {\footnotesize $0.638$}
      &
      {\footnotesize $0.994$}
      &
      {\footnotesize $43.5$}
      &
      {\footnotesize $5.06$}
      &
      {\footnotesize $-0.441$}
    \\
      {\footnotesize gm m}
      &
      {\footnotesize $1.41$}
      &
      {\footnotesize $2.24$}
      &
      {\footnotesize $43.5$}
      &
      {\footnotesize $13.8$}
      &
      {\footnotesize $-0.462$}
    \\
      {\footnotesize pg m}
      &
      {\footnotesize $1.20$}
      &
      {\footnotesize $4.88$}
      &
      {\footnotesize $43.5$}
      &
      {\footnotesize $18.6$}
      &
      {\footnotesize $-0.166$}
    \\
      {\footnotesize s m}
      &
      {\footnotesize $0.0331$}
      &
      {\footnotesize $0.0223$}
      &
      {\footnotesize $43.5$}
      &
      {\footnotesize $1.37$}
      &
      {\footnotesize $-0.723$}
    \\
      {\footnotesize t m}
      &
      {\footnotesize $0.0293$}
      &
      {\footnotesize $0.0159$}
      &
      {\footnotesize $43.5$}
      &
      {\footnotesize $0.702$}
      &
      {\footnotesize $-1.13$}
    \\
      {\footnotesize tx m}
      &
      {\footnotesize $4.66$}
      &
      {\footnotesize $7.97$}
      &
      {\footnotesize $43.5$}
      &
      {\footnotesize $36.7$}
      &
      {\footnotesize $0$}
    \\
      {\footnotesize us m}
      &
      {\footnotesize $3.87$}
      &
      {\footnotesize $1081$}
      &
      {\footnotesize $91600$}
      &
      {\footnotesize $282$}
      &
      {\footnotesize $0$}
    \\
      {\footnotesize sp m}
      &
      {\footnotesize $6.99$}
      &
      {\footnotesize $32.9$}
      &
      {\footnotesize $91,600$}
      &
      {\footnotesize $91,600$}
      &
      {\footnotesize $-0.0877$}
    \\
      {\footnotesize us a}
      &
      {\footnotesize $32.2$}
      &
      {\footnotesize $93$}
      &
      {\footnotesize $87600$}
      &
      {\footnotesize $261$}
      &
      {\footnotesize $0$}
    \\
      {\footnotesize sp a}
      &
      {\footnotesize $6.57$}
      &
      {\footnotesize $23.9$}
      &
      {\footnotesize $87,600$}
      &
      {\footnotesize $87,600$}
      &
      {\footnotesize $-0.00663$}
  \end{tabular}
  \end{center}
  \caption{Empirical results
    related to Theorems~\ref{thm:absolute1} and~\ref{thm:relative}.\label{tab:results}}
\end{table}

The numbers given in Table~\ref{tab:results} are defined as follows:
\[
  \mbox{abs factor}
  :=
  \frac
  {(S_N-S_0)^2}
  {\sum_{i=0}^{N-1}(dS_i)^2}
\]
and 
\[
  \mbox{rel factor}
  :=
  \frac
  {
    \left(
      \ln\frac{S_N}{S_0}
    \right)^2
  }
  {
    (1-\mbox{min})
    \left(
      \sum_{i=0}^{N-1}(d\ln S_i)^2
      \;\bigvee\;
      \sum_{i=0}^{N-1}
      \beta
      \left(
        \frac{dS_i}{S_i}
      \right)
    \right)
  },
\]
where
\[
  \mbox{min}
  :=
  \min_n
  \ln\frac{S_n}{S_0}.
\]
To judge the magnitude of $\mbox{abs factor}$ and $\mbox{rel factor}$
we also give the factor by which the value of the security increases
(the column ``security'')
and the factor by which the value of an index (S\&P500) increases
(the column ``index'')
over the same time period.

As we already mentioned,
our experiments implicitly assume zero interest rate,
but the results they give are roughly of the same order of magnitude
as those implied by the table on p.~376 of \citet{shiryaev:1999}.
Line~1 of that table can be interpreted
(ignoring the facts that centering is not the same thing as discounting
and that DJIA cannot be reproduced by a trading strategy)
as saying that our initial capital can be increased by a factor of roughly
\[
  12,500^{2\times0.59-1}
  \approx
  5.46
\]
in 12,500 days since 1888.

\section{Non-zero interest rate}\label{sec:interest}

Our protocols implicitly assume that the interest rate is zero.
In this section we remove this restriction.
Our protocol now involves not only security $S$ but also another security $B$
(e.g., a bank account).
Their prices are assumed positive.

\bigskip

\noindent
\textsc{The Market Protocol}

\noindent
\textbf{Players:} Investor, Market

\noindent
\textbf{Protocol:}

  \parshape=6
  \IndentI  \WidthI
  \IndentI  \WidthI
  \IndentI  \WidthI
  \IndentII \WidthII
  \IndentII \WidthII
  \IndentII \WidthII
  \noindent
  $\III_0:=1$.\\
  Market announces $S_0>0$ and $B_0>0$.\\
  FOR $n=1,2,\dots,N$:\\
    Investor announces $M_n\in\bbbr$.\\
    Market announces $S_n>0$ and $B_n>0$.\\
    $\III_n := (\III_{n-1} - M_n S_{n-1})\frac{B_n}{B_{n-1}} + M_n S_n$.

\noindent
\textbf{Additional Constraint on Market:}
Market must ensure that $S$ and $B$ are continuous.

\bigskip

\noindent
(Cf.\ the protocol and its analysis on p.~296 of \citealt{shafer/vovk:2001}.)
Intuitively, at step $n$ Investor buys $M_n$ units of $S$
and invests the remaining money in $B$,
which can be a money market account, a bond, or any other security with nonnegative prices.
The protocol of \S\ref{sec:ContFin} corresponds to a constant $B_n$.

Re-expressing Investor's capital and the price of $S$
in the \emph{num\'eraire} $B_n$, we obtain
\[
  \III^{\dagger}_n
  :=
  \III_n/B_n,
  \quad
  S^{\dagger}_n
  :=
  S_n/B_n.
\]
It is easy to see that
\[
  \III_n^{\dagger} := \III_{n-1}^{\dagger} - M_n S_{n-1}^{\dagger} + M_n S_n^{\dagger},
\]
which is exactly the expression that we had in \S\ref{sec:ContFin},
only with the daggers added.
Therefore, we can restate all results of \S\ref{sec:ContFin} for the current protocol.
For example, Theorem~\ref{thm:nonstandard} implies:
\begin{theorem}\label{thm:nonstandardbond}
  The event
  \[
    \vex (S/B) = 2
    \mbox{ or }
    S/B \approx \mbox{\rm const}
  \]
  is full.
\end{theorem}
We have an interesting all-or-nothing phenomenon:
either two securities are proportional or their ratio behaves stochastically.

\section{Continuous-time result in the drift-SDE protocol}\label{sec:ContSDE}

In this section we consider a slightly more general protocol
(see Chapter~14 of \citealt{shafer/vovk:2001}):

\bigskip

\noindent
\textsc{The Drift-SDE Protocol}

\noindent
\textbf{Players:} Forecaster, Skeptic, Reality

\noindent
\textbf{Protocol:}

  \parshape=6
  \IndentI  \WidthI
  \IndentI  \WidthI
  \IndentI  \WidthI
  \IndentII \WidthII
  \IndentII \WidthII
  \IndentII \WidthII
  \noindent
  $\III_0:=1$.\\
  Reality announces $S_0\in\bbbr$; $T_0:=S_0$.\\
  FOR $n=1,2,\dots,N$:\\
      Forecaster announces $m_n\in\bbbr$; $T_n:=T_{n-1}+m_n$.\\
      Skeptic announces $M_n\in\bbbr$.\\
      Reality announces $S_n\in\bbbr$; $x_n:=\Delta S_n$.\\
      $\III_n := \III_{n-1} + M_n (x_n-m_n)$.

\noindent
\textbf{Additional Constraint on Market:}
Market must ensure that $S$ and $T$ are continuous.

\bigskip

The main differences from the Market Protocol are that:
Market becomes Reality;
Investor becomes Skeptic;
a new player, Forecaster, is introduced,
who announces at each trial his expectation of the increment $x_n$ to be chosen by Reality
(the Market Protocol corresponds to the case where $m_n$ is always $0$).
The definition of game-theoretic upper probability is unchanged.

In Chapter~14 of \citet{shafer/vovk:2001} we describe Diffusion Protocol~1,
a game-theoretic counterpart of the standard measure-theoretic SDE
\[
  dS(t)
  =
  \mu(S(t),t) dt
  +
  \sigma(S(t),t) dW(t);
\]
this equation is modeled by Forecaster choosing the moves
\begin{equation}\label{eq:drift}
  m_n
  :=
  \mu(S_{n-1},ndt)
  dt
\end{equation}
(the \emph{drift move})
and
\[
  v_n
  :=
  \sigma^2(S_{n-1},ndt)
  dt
\]
(the \emph{volatility move}).
Already Diffusion Protocol~1 provides a flexible alternative
to the usual measure-theoretic approach to SDE;
we believe that it would be very beneficial to translate
the standard theory of SDE to the game-theoretic framework
liberating measure-theoretic results of unnecessary assumptions.
But we can also do more radical things
considering much weaker protocols than Diffusion Protocol~1.
Diffusion Protocol~2 in \citet{shafer/vovk:2001}
drops Forecaster's drift move altogether;
it turns out (\citealt{shafer/vovk:2001}, Theorem~14.1) that the Black-Scholes formula
can be proven in Diffusion Protocol~2.
(It is well known that in the measure-theoretic framework
the Black-Scholes formula does not depend on drift,
but still there is no way to drop the assumption of existence of drift.)
In this section we relax Diffusion Protocol~1 in a different way:
now we drop Forecaster's volatility move.
We will see that this will not prevent us from proving the $\sqrt{dt}$ effect.

First we motivate the conditions of our theorem.
According to~(\ref{eq:drift}), $m_n$ has the order of magnitude $dt$;
in the game-theoretic framework we also expect that the drift process
$T(t)$ will be much more stable than the process $S(t)$ itself.
Therefore, one of our conditions will be that $\sum_n m_n^2$ is infinitely small.
\begin{theorem}
  For any $\delta>0$ and $D>0$, the event
  \[
    \left.
      \begin{array}{r}
        \sum_{n=1}^N m_n^2 \approx 0 \\
        \delta < \sup_t \left|S(t)-T(t)\right| < D
      \end{array}
    \right\}
    \Longrightarrow
    \sum_{n=1}^N x_n^2
    \mbox{ is appreciable}
  \]
  has lower game-theoretic probability one.
\end{theorem}
\begin{proof}
  Set $x'_n:=x_n-m_n$.
  It is easy to see from the arguments of \S\ref{sec:ContFin} that the event
  \[
    \left.
      \begin{array}{r}
        \sum_{n=1}^N m_n^2 \approx 0 \\
        \delta < \sup_t \left|S(t)-T(t)\right| < D
      \end{array}
    \right\}
    \Longrightarrow
    \sum_{n=1}^N (x'_n)^2
    \mbox{ is appreciable}
  \]
  has lower game-theoretic probability one
  (even if the condition $\sum_{n=1}^N m_n^2 \approx 0$ is dropped).
  The fact that
  \[
    \left.
      \begin{array}{r}
        \sum_{n=1}^N m_n^2 \approx 0 \\
        \delta < \sup_t \left|S(t)-T(t)\right| < D
      \end{array}
    \right\}
    \Longrightarrow
    \sum_{n=1}^N x_n^2
    \mbox{ is limited}
  \]
  has lower game-theoretic probability one
  now follows from the closeness of $L^2$ under addition;
  more specifically, from
  \[
    x_n^2
    =
    (m_n+x'_n)^2
    \le
    2
    \left(
      m_n^2 + (x'_n)^2
    \right).
  \]
  Therefore, we only need to prove that
  \[
    \left.
      \begin{array}{r}
        \sum_{n=1}^N m_n^2 \approx 0 \\
        \delta < \sup_t \left|S(t)-T(t)\right| < D
      \end{array}
    \right\}
    \Longrightarrow
    \sum_{n=1}^N x_n^2
    \mbox{ is not infinitesimal}
  \]
  has lower game-theoretic probability one.
  In other words,
  our goal is to prove that the event
  \begin{equation}\label{eq:toprove}
    \sum_{n=1}^N m_n^2 \approx 0
    \;\;\&\;\;
    \delta < \sup_t \left|S(t)-T(t)\right| < D
    \;\;\&\;\;
    \sum_{n=1}^N x_n^2 \approx 0
  \end{equation}
  has zero upper game-theoretic probability.

  According to~(\ref{eq:change1}), we have, for some strategy $S_1$ for Skeptic,
  \[
    \III^{S_1}_N - \III^{S_1}_0
    =
    C (S_N-T_N)^2
    -
    C
    \sum_{i=1}^N
    (x'_i)^2.
  \]
  Since
  \[
    x_i^2
    =
    (x'_i+m_i)^2
    =
    (x'_i)^2 + 2 x'_i m_i + m_i^2,
  \]
  we can rewrite this equality as
  \[
    \III^{S_1}_N - \III^{S_1}_0
    =
    C (S_N-T_N)^2
    -
    C
    \sum_{i=1}^N
    x_i^2
    +
    C
    \sum_{i=1}^N
    m_i^2
    +
    \III^{S_2}_N - \III^{S_2}_0,
  \]
  where $S_2$ is Skeptic's strategy that recommends move $2Cm_i$ at trial $i$.
  Therefore, there is Skeptic's strategy that ensures
  \[
    \III_N - \III_0
    =
    C (S_N-T_N)^2
    -
    C
    \sum_{i=1}^N
    x_i^2
    +
    C
    \sum_{i=1}^N
    m_i^2,
  \]
  and we can take $\III_0$ to be 1.
  On the event~(\ref{eq:toprove}) this strategy
  (if stopped at the first moment that $|S(t)-T(t)|>\delta$)
  will attain at least a capital of $C\delta^2$,
  which can be made as large as we wish
  by choosing a large $C$.
  \qed
\end{proof}

\section{A modified Market Protocol}\label{sec:alternative}

To state Theorems~\ref{thm:nonstandardlow} and~\ref{thm:nonstandardhigh} in a nicer way
(avoiding the $\epsilon$ and $D$),
we change the Market Protocol in the following way.
The two parameters of the Market Protocol were $T$, the time horizon,
and $N$, the infinite number of subintervals into which the interval $[0,T]$ was split.
Now we allow $T$ to be an infinitely large positive number
(still requiring $dt:=T/N$ to be infinitesimal)
and add another parameter,
an infinitely small positive number $\epsilon$.
(Of course, $T$ can stay limited if we wish.)
The Additional Constraint on Market is now changed to
``Market must ensure that $\sup|\Delta S|\le\epsilon$''.
The upper probability $P$ in this protocol is defined by the formula
\[
  P(E)
  :=
  \inf
  \left\{
    \III^S(\Box)
    \st
    \inf_{0\le t\le T}\III^S(t)\ge0 \mbox{ everywhere},\enspace
    \III^S(T)\ge1 \mbox{ inside } E
  \right\},
\]
where $S$ ranges over (internal) strategies;
the expressions such as ``almost certain'' refer to this upper (and the corresponding lower) probability.
Remember that a hyperreal number $t$ is \emph{appreciable} if $a<|t|<b$
for some positive real $a$ and $b$
(i.e., if it is neither unlimited nor infinitesimal).
\begin{theorem}\label{thm:nonstandard}
  It is almost certain that
  \[
    \sup_t
    \left|
      S(t) - S(0)
    \right|
    \mbox{ is appreciable }
    \Longrightarrow
    \vex S = 2.
  \]
  More precisely,
  \begin{equation}\label{eq:low}
    \vex S < 2
    \Longrightarrow
    \sup_t
    \left|
      S(t) - S(0)
    \right|
    \mbox{ is infinitesimal}
  \end{equation}
  and
  \begin{equation}\label{eq:high}
    \vex S > 2
    \Longrightarrow
    \sup_t
    \left|
      S(t) - S(0)
    \right|
    \mbox{ is unlimited}.
  \end{equation}
\end{theorem}
\begin{proof}
  Of course, the proof is a modification of the proofs
  of Theorems~\ref{thm:nonstandardlow} (for (\ref{eq:low}))
  and~\ref{thm:nonstandardhigh} (for (\ref{eq:high}));
  we again assume $S(0)=0$.

  First we prove~(\ref{eq:low}).
  As before, we consider the strategies $M_n^{(C)}:=2CS_n$
  starting from the initial capital $1$,
  with the only difference that the strategy stops playing
  (i.e., starts choosing the move $0$)
  as soon as $C\sum_{i=0}^{n-1}(dS_i)^2$ reaches the value $1-C\epsilon^2$
  (in particular, the strategy never plays if $C\epsilon^2\ge1$;
  this stopping rule ensures that the strategy never goes bankrupt)
  or $|S_n|>C^{-1/2}$ (this condition replaces $|S_n|>\delta$),
  whichever happens earlier.
  Now we can combine these strategies into
  \[
    M_n
    :=
    \sum_{m=1}^{\infty}
    2^{-m}
    M_n^{(2^m)}
  \]
  (we do not have any problems of convergence since for each standard $\epsilon>0$
  only finitely many strategies $M_n^{(2^m)}$ will ever play).
  It is clear that this strategy will ensure an unlimited final capital $\III_N$.

  It remains to prove~(\ref{eq:high}).
  Consider the strategies $M_n^{(D)}:=-2D^{-1/2}S_n$
  starting from the initial capital $1$,
  with the only difference that the strategy stops playing
  as soon as $D^{-2}S_n^2$ reaches the value $1-D^{-2}(2S_n\epsilon+\epsilon^2)$.
  This way we make sure that the strategy never goes bankrupt.
  Combining, as before, the strategies $M_n^{(D)}$ into
  \[
    M_n
    :=
    \sum_{m=1}^{\infty}
    2^{-m}
    M_n^{(2^m)}
  \]
  (the convergence follows from $\sum_{m=1}^{\infty}2^{-3m}<\infty$),
  we can see that the combined strategy will ensure an unlimited final capital $\III_N$.
  \qed
\end{proof}

\subsection*{Acknowledgements}

The theorems in \S~\ref{sec:ContFin} were inspired by questions raised by Freddy Delbaen
in his review of our book in the \emph{Journal of the American Statistical Association}.
Our reply to his review is available at \url{www.probabilityandfinance.com}.

\addcontentsline{toc}{section}{\refname}
\newcommand{\noopsort}[1]{}

\appendix
\section{Continuous games}\label{app:contgame}

This appendix contains a partial account of the concepts from nonstandard analysis
used in this article.
(It was intended as an improvement over Appendix~11.5 of~\citet{shafer/vovk:2001}.)

The continuous games that we consider in the main part of the article are ultraproducts of discrete games.
We will first explain informally how such ultraproducts are formed.
For a formal exposition of the concept of an ultraproduct,
the reader may consult \citet{eklof:1977} or the classical article by Jerzy \citet{los:1955}.

In general, an ultraproduct is formed from 
a sequence $\OOO_1,\OOO_2,\dots$ of similar mathematical structures,
perhaps identical or perhaps increasing in size.
We remain informal by not saying what
we mean by ``similar'', but the idea is that certain statements
have a meaning in each of the $\OOO_n$.  A statement that two objects
are related in a certain way, for example, might be interpreted 
in $\OOO_n$ as $R_n(x_n,y_n)$, where $x_n$ and $y_n$
are objects in $\OOO_n$ and $R_n$ is a binary relation in $\OOO_n$.
Such a statement should also have a reference $R(x,y)$ in the ultraproduct.
Intuitively, 
\begin{itemize}
  \item
    $R$ is the sequence $R_1,R_2,\dots$, 
  \item
    $x$ is the sequence $x_1,x_2,\dots$,
  \item
    $y$ is the sequence $y_1,y_2,\dots$, and 
  \item
    $R(x,y)$ holds if $R_n(x_n,y_n)$ holds for most $n$.
\end{itemize}
To make ``most'' precise, we choose a nontrivial ultrafilter in $\bbbn$,
the set of natural numbers (positive integers).  
A nontrivial ultrafilter $\UUU$ in $\bbbn$
is a set of subsets of $\bbbn$ that has, inter alia, the property that whenever we partition $\bbbn$ into two sets,
exactly one of the two sets is in $\UUU$.
We say a relation holds for most $n$ if the set of $n$ for which it holds is in $\UUU$.

To strengthen this explanation, we now review the concept of 
an ultrafilter and provide two examples of an ultraproduct:  (i) the hyperreals,
and (ii) a simple continuous game.

\subsection{Ultrafilters}

An \emph{ultrafilter}
in $\bbbn$ is a family $\UUU$ of subsets of $\bbbn$ such that
\begin{enumerate}
\item
  $\bbbn\in\UUU$ and $\emptyset\notin\UUU$,
\item
  if $A\in\UUU$ and $A\subseteq B\subseteq\bbbn$, then $B\in\UUU$,
\item
  if $A\in\UUU$ and $B\in\UUU$, then $A\cap B\in\UUU$, and
\item
  if $A \subseteq \bbbn$,
  then either $A \in \UUU$ or $\bbbn \setminus A \in \UUU$.
\end{enumerate}
(The first three properties define a \emph{filter}.)
An ultrafilter $\UUU$ is \emph{nontrivial} 
if it does not contain a set consisting of a single integer;
this implies that all the sets in $\UUU$ are infinite.
It follows from the axiom of choice that a nontrivial ultrafilter exists.
We fix a nontrivial ultrafilter $\UUU$.

We say that a property of natural numbers holds
for \emph{most} natural numbers
(or for most $k$, as we will say for brevity)
if the set of natural numbers for which it holds is in $\UUU$;
Condition~2 of the definition justifies this usage.
It follows from Condition~4 that for any property $A$,
either $A$ holds for most $k$ or else the negation of $A$ holds for most $k$.
It follows from Conditions~1 and~3
that $A$ and its negation cannot both hold for most $k$.

\subsection{The hyperreals}

As a first example of an ultraproduct, we construct the hyperreals,
as they are usually constructed
in nonstandard analysis~\citet{goldblatt:1998}.
In this case, the objects $\OOO_n$ are all identical---each is a copy
of the real numbers,
together with the usual operations and relations associated with them.

As a first approximation, a \emph{hyperreal number}  $a$
is a sequence $\left[a^{(1)}a^{(2)}\dots\right]$ of real numbers.
Sometimes we abbreviate $\left[a^{(1)}a^{(2)}\dots\right]$ to $\left[a^{(k)}\right]$.
Operations (addition, multiplication, etc.)\ over hyperreals
are defined term by term\label{p:stanfunc}.
For example,
\[
  \left[a^{(1)} a^{(2)} \dots\right]
  +
  \left[b^{(1)} b^{(2)} \dots\right] :=
  \Bigl[ \left(a^{(1)}+b^{(1)}\right) \left(a^{(2)}+b^{(2)}\right) \dots \Bigr].
\]
Relations (equals, greater than, etc.)\ are extended to the hyperreals by voting.
For example,
$[a^{(1)} a^{(2)} \dots] \le [b^{(1)} b^{(2)} \dots]$
if $a^{(k)} \le b^{(k)}$ for most $k$.
For all $a,b \in \ns\bbbr$
one and only one of the following three possibilities holds:
$a<b$, $a=b$, or $a>b$.

Perhaps we should dwell for a moment on the fact that
a hyperreal number $a = \left[a^{(1)}a^{(2)}\dots\right]$
is always below, equal to, or above
another hyperreal number $b = \left[b^{(1)}b^{(2)}\dots\right]$:
$a<b$, $a=b$, or $a>b$.
Obviously some of the $a^{(k)}$ can be above $b^{(k)}$,
some equal to $b^{(k)}$, and some below $b^{(k)}$.
But the set of $k$ satisfying one these three conditions is in $\UUU$
and outvotes the other two.

We do not distinguish hyperreals $a$ and $b$ such that $a=b$.
Technically, this means that a hyperreal
is an equivalence class of sequences rather than an individual sequence:
$[a^{(1)} a^{(2)} \dots]$ is the equivalence class containing $a^{(1)} a^{(2)} \dots$\,.

We embed the real numbers in the hyperreals 
by identifying each real number $a$ with $[a,a,\dots]$.
For each $A\subseteq\bbbr$ we denote by $\ns A$
the set of all hyperreals $[a^{(k)}]$
with $a^{(k)}\in A$ for all $k$.  We call $\ns \bbbn$ the \emph{hypernaturals}.

We say that $a \in \ns\bbbr$ is \emph{infinitesimal}
if $\left|a\right| < \epsilon$ for each real $\epsilon > 0$.
The only real number that qualifies as an infinitesimal by this definition 
is $0$\label{p:zeroinf}.
We say that $a \in \ns\bbbr$ is \emph{infinitely large}
if $a > C$ for each positive integer $C$,
and we say that $a \in \ns\bbbr$ is \emph{finite}
if $a < C$ for some positive integer $C$.

We write $a \approx b$ when $a - b$ is infinitesimal.
For every hyperreal number $a\in\ns\bbbr$
there exists a unique standard number
$\standard(a)$ (its \emph{standard part})
such that $a\approx b$.

The representation of the hyperreals as equivalence classes
of sequences with respect to a nontrivial ultrafilter is 
constructive only in a relative sense, because the proof that
a nontrivial ultrafilter exists is nonconstructive;
no one knows how to exhibit one.
However, the representation provides
an intuition that helps us think about hyperreals.
For example, an infinite positive integer is represented by
a sequence of positive integers that increases without bound,
such as $[1,2,4,\dots]$, and the faster it grows the larger it is.

\subsection{An ultraproduct of games}

Now we construct a continuous game, of the type used in this article.

In this construction, the following protocol, where $n$ is a natural number:

\bigskip

\noindent
\textbf{Protocol:}

  \parshape=6
  \IndentI  \WidthI
  \IndentI  \WidthI
  \IndentI  \WidthI
  \IndentII \WidthII
  \IndentII \WidthII
  \IndentII \WidthII
  \noindent
  $\III_0:=1$.\\
  Market announces $S_0\in\bbbr$.\\
  FOR $n=1,2,\dots,N$:\\
      Investor announces $M_n\in\bbbr$.\\
      Market announces $S_n\in\bbbr$.\\
      $\III_n := \III_{n-1} + M_n \Delta S_n$.

\bigskip

\noindent
Fix a positive real number $T$ and an infinitely large positive integer $N$;
let
\[
  N = [N^{(k)}] = [N^{(1)},N^{(2)},\dots].
\]
For each natural number $k$, set
\[
  \bbbt^{(k)} := \{ nT/N^{(k)} \st n=0,1,\dots,N^{(k)} \}.
\]
To each $k$ corresponds a ``finitary framework''
(which we will call the $k$-finitary framework),
where the time interval is the finite set $\bbbt^{(k)}$
rather than the infinite set $\bbbt$.
The ``limit'' (formally, ultraproduct) of these finitary frameworks will be
the infinitary framework based on $\bbbt$;
as in the previous subsection,
this ``limit'' is defined as follows:
\begin{itemize}
\item
  An object in the infinitary framework,
  such as strategy,
  should be defined as a family of finitary objects:
  for every $k$, an object in the $k$-finitary framework should be defined
  (cf.\ the definition of hyperreals in the previous subsection).
\item
  Functionals defined on finitary objects
  are extended to infinitary objects term-wise,
  analogously to the previous subsection.
  (By ``functionals'' we mean functions of objects of complex nature,
  such as paths or strategies.)
\item
  Relations (in particular, properties) are defined by voting
  (again as in the previous subsection).
\end{itemize}
(In nonstandard analysis
such limiting infinitary structures are called hyperfinite.)

\subsection{Details of the proof of Theorem~\ref{thm:nonstandardlow}}

Let us show more formally why $\III_n$ is nonnegative and why $\III_N\ge C\delta^2$.

According to the first equality in~(\ref{eq:change1}), in every finitary framework we have
\[
  \III_n - 1
  \ge
  -
  C
  \sum_{i=0}^{N-1}
  (dS_i)^2;
\]
since the value on the right-hand side is infinitesimal
(and, therefore, smaller than $1$ in absolute value),
$\min_n\III_n$ is positive.

To see that $\III_N\ge C\delta^2$,
define in each finitary framework the stopping time
\[
  n
  :=
  \min
  \left\{
    i
    \st
    |S_i| > \delta
  \right\}.
\]
Again using the first equality in~(\ref{eq:change1}) we obtain that in each finitary framework
\[
  \III_N - 1
  =
  \III_n - 1
  =
  C S_n^2
  -
  C
  \sum_{i=0}^{n-1}
  (dS_i)^2
  >
  C \delta^2
  -
  C
  \sum_{i=0}^{N-1}
  (dS_i)^2;
\]
it remains to remember that the last subtrahend is infinitely small
and, therefore, smaller than $1$.
\end{document}